\documentclass[12pt]{amsart}
\usepackage[utf8]{inputenc}
\usepackage[T1]{fontenc}

\usepackage[margin=2.5cm]{geometry}
\usepackage{amsthm, amsmath, amssymb}
\usepackage{float}
\usepackage{xcolor}
\usepackage{comment}
\usepackage{thm-restate}
\usepackage{enumerate}
\usepackage{hyperref}
\usepackage{graphicx}

\usepackage{subcaption}
\usepackage{ifthen}

\usepackage[color=red!30]{todonotes}

\hypersetup{
  pdftitle = {Kick the cliques},
  pdfauthor = {G. Berthe, M. Bougeret, D. Gonçalves, J.-F. Raymond},
  colorlinks = true,
  linkcolor = black!30!red,
  citecolor = black!30!green
}
\usepackage{thmtools}
 \declaretheorem[name=Theorem, numberwithin=section]{theorem}
 \declaretheorem[name=Lemma, sibling=theorem]{lemma}

 \declaretheorem[name=Corollary, sibling=theorem]{corollary}

 \declaretheorem[name=Claim, sibling=theorem]{claim}
\newcommand{\Pref}[1]{\hyperref[#1]{P\ref*{#1}}}

\newcommand{\N}{\mathbb{N}}
\newcommand{\R}{\mathbb{R}}
\newcommand{\cB}{\mathcal{B}}

\newcommand{\cD}{\mathcal{D}}

\newcommand{\cK}{\mathcal{K}}

\newcommand{\cP}{\mathcal{P}}
\newcommand{\cY}{\mathcal{Y}}
\newcommand{\cZ}{\mathcal{Z}}

\def\cqedsymbol{\ifmmode$\lrcorner$\else{\unskip\nobreak\hfil
\penalty50\hskip1em\null\nobreak\hfil$\lrcorner$
\parfillskip=0pt\finalhyphendemerits=0\endgraf}\fi}

\newcommand{\ceil}[1]{\left\lceil#1\right\rceil}
\newcommand{\intv}[2]{\left \{ #1, \dots, #2 \right \}}
\newcommand{\cov}{\textsc{Cover}}
\newcommand{\trh}{\textsc{Triangle Hitting}}

\newcommand{\ann}{\textsc{Ann.}}
\newcommand{\eps}{\varepsilon}
\DeclareMathOperator{\ext}{\sf Petals}
\DeclareMathOperator{\tw}{\bf tw}

\title{Kick the cliques}

\author[G.~Berthe]{Gaétan Berthe}
\author[M.~Bougeret]{Marin Bougeret}
\author[D.~Gonçalves]{Daniel Gonçalves}
\address[G.~Berthe, M.~Bougeret, D.~Gonçalves]{LIRMM, Université de Montpellier, CNRS, Montpellier, France.}

\author[J.-F.~Raymond]{Jean-Florent Raymond}
\address[J.-F.~Raymond]{CNRS, LIP, ENS de Lyon, France.}
\email{jean-florent.raymond@cnrs.fr}

\date{\today}

\begin{document}
\maketitle

\begin{abstract}
In the $K_r$-\textsc{Cover} problem, given a graph $G$ and an integer $k$ one has to decide if there exists a set of at most $k$ vertices
whose removal destroys all $r$-cliques of $G$.

In this paper we give an algorithm for $K_r$-\textsc{Cover} that runs in subexponential FPT time on graph classes satisfying two simple conditions related to cliques and treewidth.
As an application we show that our algorithm solves $K_r$-\textsc{Cover} in time
\begin{itemize}
    \item $2^{O_r\left (k^{(r+1)/(r+2)}\log k \right)} \cdot n^{O_r(1)}$ in pseudo-disk graphs and map-graphs;
    \item $2^{O_{t,r}(k^{2/3}\log k)} \cdot n^{O_r(1)}$ in $K_{t,t}$-subgraph-free string graphs; and
    \item $2^{O_{H,r}(k^{2/3}\log k)} \cdot n^{O_r(1)}$ in $H$-minor-free graphs.
\end{itemize}
\end{abstract}

\section{Introduction}

In the $K_r$-\textsc{Cover} problem, given a graph $G$ and an integer $k$ one has to decide if there are $k$ vertices in $G$ whose deletion yields a $K_r$-free graph. This problem falls within the general family of (implicit) covering problems (also called hitting problems) and encompasses several extensively studied problems 
such as the case $r=2$ better known under the name \textsc{Vertex Cover} and the case $r=3$ that we usually refer to as \trh{}.
Already for these small values the problem is NP-complete.

In this paper we are interested in \emph{subexponential parameterized algorithms} for $K_r$-\cov{}, i.e., algorithms that run in \emph{Fixed parameter Tractable (FPT)} time (that is, time $f(k)\cdot n^{O(1)}$ for some computable function $f$) and where additionally the contribution of the parameter $k$ is subexponential, in other words $f(k) \in 2^{o(k)}$.
Under the Exponential Time Hypothesis of Impagliazzo and Paturi \cite{impagliazzo2001complexity}, such algorithms do not exist in general and so we have to focus on particular graph classes. 

Historically, subexponential graph algorithms were first obtained for specific problems in sparse graph classes such as planar graphs.
The techniques used have then been unified and extended by Demaine,  Hajiaghayi, and Thilikos in the meta-algorithmic theory of \emph{bidimensionality}~\cite{demaine2005subexponential}, which provides a generic machinery to solve a wide range of problems in subexponential FPT time on $H$-minor free graphs.
Initially bidimensionality was defined for graph classes with some ``flatness'' property similar to planar graphs, typically graphs of bounded genus and $H$-minor-free graphs. Over the years, the theory saw several improvements and extensions in order to deal with different settings like map graphs and other classes of intersection graphs, which are initially not sparse as they contain large cliques for example, but where we can branch in subexponential-time to reduce to sparse instances (see for instance the bibliography cited in~\cite{berthe24}). 
However, despite its generality, bidimensionality can only handle the so-called \emph{bidimensional problems} where, informally, as soon as the instance $(G,k)$ 
 contains a large $t\times t$ grid as a minor (for $t \in o(k)$, typically $t=\sqrt{k}$), we know that $(G,k)$ is necessarily a no-instance (or yes-instance depending on the problem).
 This is the case of \textsc{Vertex Cover} but unfortunately not of \trh{} (as grids are triangle-free) and more generally not of $K_r$-\cov{} for $r\geq 3$.

The focus of this paper is on this blind spot: subexponential FPT algorithms for a problem that is not bidimensional, namely $K_r$-\cov{}.
About this problem, we note that using arguments developed in the context of approximation~\cite{fomin2011bidimensionality}, the following subexponential FPT algorithm can be obtained for apex-minor free graphs (which are sparse).

\begin{theorem}[from \cite{fomin2011bidimensionality}]\label{th:apexfree}
For every apex\footnote{A graph is \emph{apex} if the deletion of some vertex yields a planar graph.} graph $H$ and every $r\in \N$ there is an algorithm solving $K_r$-\cov{}\footnote{Actually the statement also applies to $F$-\cov{} for any $F$.} in $H$-minor-free graphs in time 
$
2^{O_{H,r}\left (\sqrt{k} \right )} \cdot n^{O_{r}(1)}.
$
\end{theorem}

Regarding classes that are not sparse, \trh{} received significant attention in the last years in classes of intersection graphs such as (unit) disk graphs\footnote{\emph{(Unit) disk graphs} are intersection graphs of (unit) disks in $\R^2$.}, pseudo-disk graphs, and subclasses of segment graphs\footnote{\emph{Segment graphs} are intersection graphs of segments in $\R^2$.} \cite{lokSODA22, Faster2023An, berthe24, berthe24bFVS}. We only recall below the most general results and do not mention those that require a geometric representation.

\begin{theorem}\label{th:pos}
There are algorithms that given a parameter $k$ and a $n$-vertex graph (without a geometric representation) solve \trh{} in time
\begin{enumerate}
    \item \label{e:oh} $2^{O(k^{3/4} \log k)}n^{O(1)}$ in disk graphs~\cite{Faster2023An};\footnote{The published version of the paper gives a bound of $2^{O(k^{4/5} \log k)}n^{O(1)}$ but it can easily be improved to $2^{O(k^{3/4} \log k)}n^{O(1)}$, as confirmed to us by the authors of \cite{Faster2023An} (private communication).}
    \item \label{e:contact} $2^{O(k^{3/4}\log k)}n^{O(1)}$ in contact-segment\footnote{\emph{Contact-segment graphs} are the intersection graphs of non-crossing segments in $\R^2$.} graphs~\cite{berthe24};
    \item \label{e:ddir} $2^{O_{t,d}(k^{2/3})\log k} n^{O(1)}$in $K_{t,t}$-subgraph-free $d$-DIR\footnote{A graph is \emph{$d$-DIR} if it is the intersection graph of segments of $\R^2$ with at most $d$ different slopes.} graphs~\cite{berthe24}.
\end{enumerate}
\end{theorem}

\begin{theorem}[\cite{berthe24}]\label{th:neg}
Assuming the Exponential Time Hypothesis, there is no algorithm solving \trh{} in time
\begin{enumerate}
    \item $2^{o(n)}$ in 2-DIR graphs;
    \item $2^{o(\sqrt{\Delta n})}$ in 2-DIR graphs with maximum degree $\Delta$; and
    \item $2^{o(\sqrt{n})}$ in $K_{2,2}$-free contact-2-DIR graphs of maximum degree~6.
\end{enumerate}
\end{theorem}

\subsection*{Our contribution}
Our main result is the following subexponential parameterized algorithm for $K_r$-\cov{} in graph classes satisfying two conditions related to cliques and treewidth. Notice that the statement of the following theorem is a simplified version of the actual \autoref{th:kick} that we prove in \autoref{sec:kick}.

\begin{theorem}\label{th:kick_fast}
    Let $r\in \N$, $\alpha\in (0,1)$, $\mu\in \R_{>0}$ and let $\mathcal{G}$ be a hereditary graph class where every $G\in \mathcal{G}$ with $n$ vertices and clique number $\omega$ has $O_r(\omega^\mu n)$ cliques of order less than $r$ and treewidth $O_r(\omega^\mu n^\alpha)$.
    There exists $\eps < 1$ and an algorithm that solves $K_r$-\cov{} on $\mathcal{G}$ in time
    $2^{k^\eps} \cdot n^{O_r(1)}.$
\end{theorem}
One additional motivation for this work was to generalize to $K_r$-\cov{} the techniques used in previous work to solve \trh{} (in specific graph classes) and to extract the minimal requirements for such an approach to work in more general settings. We believe we met this goal as we actually describe a single generic approach that solves $K_r$-\cov{} on any input graph, for any~$r$. The properties of the class in which the inputs are taken is only used to bound its running time. Such a generalization effort can be fruitful and indeed it allowed us afterwards to identify natural graph classes where subexponential algorithms exist as a consequence of our general result.

In \autoref{sec:apps} we derive from \autoref{th:kick} the following applications.

\begin{restatable}{theorem}{thpseudo}\label{th:pseudo}
There is an algorithm solving $K_r$-\cov{} in pseudo-disk graphs in time 
$
2^{O_r\left (k^{(r+1)/(r+2)}\log k \right )} \cdot n^{O_r(1)}.
$
\end{restatable}
\emph{Pseudo-disk graphs} are a classical generalization of disk graphs where to each vertex is associated a \emph{pseudo-disk} (a subset of the plane that is homeomorphic to a disk), two vertices are adjacent if the corresponding pseudo-disks intersect and additionally we require that for any two intersecting pseudo-disks, their boundaries intersect on at most two points. Disk graphs and contact segment graphs are pseudo-disk graphs, so \autoref{th:pseudo} applies to the two settings handled by the algorithms of \cite{Faster2023An} and \cite{berthe24} mentioned at items \ref{e:oh} and~\ref{e:contact} of \autoref{th:pos}. Another application is the following:

\begin{restatable}{theorem}{thmap}\label{th:map}
There is an algorithm solving $K_r$-\cov{} in map graphs\footnote{\emph{Map graphs} are intersection graphs of interior-disjoint regions of $\R^2$ homeomorphic to disks.} in time 
$
2^{O_r\left (k^{(r+1)/(r+2)}\log k \right )} \cdot n^{O_r(1)}.
$
\end{restatable}

We cannot expect a similar consequence for the more general class of string graphs.\footnote{\emph{String graphs} are intersectiong graphs of Jordan arcs in $\R^2$. They generalize many of the most studied classes of intersection graphs of geometric objects in the plane such as disk graphs, pseudo-disk graphs, segment graphs, chordal graphs, etc.} Indeed, there are $n$-vertex string graphs that are triangle-free and have treewidth $\Omega(n)$, for instance the balanced bicliques.\footnote{$K_{n,n}$ can be drawn as a 2-DIR graph with $n$ horizontal disjoint segments that are all crossed by $n$ vertical disjoint segments.} Note that such graphs prevent string graphs to satisfy the requirement of \autoref{th:kick_fast}. Also, and more importantly, by \autoref{th:neg} under ETH there is no $2^{o(n)}$-time algorithm for $K_3$-\cov{} in 2-DIR graphs, a restricted subclass of string graphs.
As we will show, large bicliques are the only obstructions in the sense that forbidding them in string graphs allows us to solve the problem in subexponential FPT time.
For this we use the following light version of \autoref{th:kick_fast} (also consequence of \autoref{th:kick}) suited for classes where the clique number is already bounded.
\begin{theorem}\label{th:kick_weakly}
    Let $r\in \N$, $\alpha\in (0,1)$ and let $\mathcal{G}$ be a hereditary graph class where every $G\in \mathcal{G}$ with $n$ vertices has $O(n)$ cliques of order less than $r$ and treewidth $O(n^\alpha)$.
    There exists an algorithm that solves $K_r$-\cov{} on $\mathcal{G}$ in time
    $
    2^{O_r(k^{2/(1+1/\alpha)} \log k)} \cdot n^{O_r(1)}.
    $
\end{theorem}

As a consequence we obtain a subexponential FPT algorithm for string graphs excluding large bicliques.
\begin{restatable}{theorem}{thstring}\label{th:string}
There is an algorithm solving $K_r$-\cov{} in $K_{t,t}$-subgraph-free string graphs in time 
$2^{O_{t,r}(k^{2/3}\log k)} \cdot n^{O_r(1)}.$
\end{restatable}
\autoref{th:string} is a generalization in two directions (the objects to cover and the graph to consider) of item~\ref{e:ddir} of \autoref{th:pos}. We note that under ETH the contribution of $k$ cannot be improved to $2^{o(\sqrt{k})}$, according to \autoref{th:neg}.
Finally we observe that \autoref{th:kick_fast} can also be applied to certain classes of sparse graphs.

\begin{restatable}{theorem}{thminorfree}\label{th:minorfree}
For every graph $H$, there is an algorithm solving $K_r$-\cov{} in $H$-minor-free graphs in time 
$
2^{O_{H,r}(k^{2/3}\log k)} \cdot n^{O_r(1)}.
$
\end{restatable}
This is a more general statement than \autoref{th:apexfree} in the sense that we are not limited to apex-minor free graphs, to the price of a slightly larger time complexity.

\subsection*{Our techniques}
Our subexponential algorithm for $K_r$-\cov{} of \autoref{th:kick_fast} is obtained as follows. Given $(G,k)$, we first perform in \autoref{sec:covers} a preliminary branching step  whose objective is to get rid of large cliques (i.e., cliques of order at least $k^\eps$ for some  $\eps  \in (0,1)$ that we will fix later).
This step is a folklore technique which is frequently used for any problem where a solution has to contain almost all vertices of a large clique, like \trh{} (or in general $K_r$-\cov{}), \textsc{Feedback Vertex Set}, or \textsc{Odd Cycle Transversal}. After this preprocessing has been performed we can assume that the instances to solve have no clique on more than $k^\eps$ vertices.
Then, we greedily compute an $r$-approximate $K_r$-cover $M$. If $|M| >kr$ we can already answer negatively, so in the following we may assume $|M| \le kr$ and will use the fact that there is no $K_r$ in $G - M$. 

Now, the crucial part of the algorithm is \autoref{sec:petals}.
Informally, the goal of the algorithms described in this section is to extend $M$ into a superset $M'$ together with a new parameter $k'\leq k$, such that $|M'|=O\left (k^{1+c\eps} \right )$ (for some $c>0$) and that vertices of $V(G)\setminus M'$ are irrelevant for the problem of covering $r$-cliques.
In that way, we can remove them, and it remains only to solve $(G[M'],k')$, whose treewidth can be typically bounded by $\sqrt{\omega(G)|M'|}$ in the graph classes we consider.
As $\sqrt{\omega(G)|M'|}=O\left (k^{1/2+\frac{c+1}{2}\eps} \right )$, this leads to a subexponential algorithm. We stress that the high-level description above is not a kernelization because first we actually do not produce a single reduced instance but instead we have to branch and obtain a subexponential (in $k$) number of sub-instances and second because the reduction steps are not carried out in polynomial time, but in subexponential (in $k$) time.

To obtain this set $M'$, we use lemmatas~\ref{lem:flower-picking} and~\ref{lem:flowerS-picking} that are inspired from the following ``virtual branching'' procedure of \cite[Lemma~6.5]{lokSODA22}.
This routine was introduced for \trh{} and works as follows.
It starts with a triangle hitting set $M$ (obtained by greedily packing disjoint triangles), and outputs a slightly larger superset $M'$ 
such that vertices in $G-M'$ are \emph{almost} useless, in the sense that every triangle has at least two vertices in $M'$ (we do not detail here how are handled the triangles with exactly two vertices in $M'$ and refer to \cite{lokSODA22}). This is done as follows.
For a vertex $v \in M$, consider a maximum matching $M(v) \subseteq N(v) \cap (V(G) \setminus M)$. If for every $v\in M$ such matching is small, meaning $|M(v)| \le k^\eps$, then we can define $M' = M \cup \bigcup_{v \in M}M(v)$. We are done as $|M'|=O(k^{1+\eps})$ remains small, and there is no longer a $v \in M'$ with an edge in $N(v) \setminus M'$ (as this would form a triangle outside $M$).
Otherwise, if for some $v\in M$, $|M(v)| > k^\eps$, a solution of \trh{} should either take $v$, or otherwise covers all edges of $M(v)$. In the second case, it would be too costly to guess which vertex is taken in each $e \in M(v)$, so instead the procedure ``absorbs'' $M(v)$ by defining $M' = M \cup M(v)$. This absorption increases the size of $M$, but ``virtually'' decreases the parameter $k$, as it increases by $|M(v)|$ the size of a matching that the solution will have to cover. This leads to a running time typically dominated by the recurrence $f(k)=f(k-1)+f(k-k^\eps)$, which is subexponential in $k$.

Now, coming back to $K_r$-\cov{}, given a set $M$, let us say\footnote{This notion of type-$i$ clique will not be used later and is just introduced for this sketch.} that a \emph{type-$i$ clique} is an $r$-clique $X$ such that $|X \cap M|=i$.
We could remove type-1 cliques by using the previous virtual branching procedure, defining now $M(v)$ as a packing of $r-1$ cliques instead of a packing of edges, but the problem is that, if we want to obtain a set $M'$ as promised (where vertices of $V(G)\setminus M'$ are useless), we also have to remove type-$i$ cliques for $i \in \intv{2}{r-1}$.
However, there is a first obstacle to remove such type-$i$ cliques: as the part common with $M$ (which was before a single vertex in $M$) is now an $i$-clique, we cannot afford to enumerate all possible choices $X'$ of such an $i$-clique in $M$. Indeed, already for $i=2$ we would possibly consider a quadratic number of sets $X'$, so absorbing 
every packing (of $(r-i)$-cliques) $M(X')$  (in the unfortunate case where these are all small) would result in a set $M'=M \cup \bigcup_{X' \subseteq M, |X'|=2}M(X')$ with $|M'| =  \Omega (k^2)$, which is too large for our purpose.
To circumvent this issue, we identified a key property that holds in many geometric graph classes like pseudo-disk or $K_{t,t}$-subgraph-free string graphs: in such graphs, there is only a linear (in the number of vertices) number of $i$-cliques for fixed $i$, and for fixed clique number~$\omega$. In our case, as $\omega$ is small, and $i \le r$ is fixed, this implies that there are $O(|M|)$ such $i$-cliques in $M$. Hence, it allows us to control the size of $M'$.
Of course, dealing with $r$-cliques instead of triangles also raises other problems, in particular related to the way we cover their intersection with $M$ which is no longer a single vertex but a clique. To deal with this issue we had to introduce an annotated variant of the problem where additional sets of vertices have to be covered besides $r$-cliques.

\subsection*{Organization of the paper}
In \autoref{sec:prelim} we give the necessary definitions. We describe the first branching in \autoref{sec:covers} and the second in \autoref{sec:petals}. The algorithm is given in \autoref{sec:kick}. We give applications to selected graph classes in \autoref{sec:apps}.
We conclude with open questions in \autoref{sec:open}.

\section{Preliminaries}
\label{sec:prelim}

\subsection*{Running times}
When stating results related to algorithms, the variable $n$ in the running time always refers to the number of vertices of the graph that is part of the input.
To simplify the presentation we will assume that $r$ is a fixed constant instead of explicitly give it as a parameter in all our algorithms and lemmas.

\subsection*{Graphs}
Unless otherwise stated we use standard graph theory terminology. A \emph{clique} in a graph $G$ is a complete subgraph and when there is no ambiguity we also use \emph{clique} to denote a subset of $V(G)$ inducing a complete subgraph.
The \emph{clique number} of $G$ is the maximum number of vertices of a clique it contains and we denote it by $\omega(G)$.
For any $i\in\N$, an \emph{$i$-clique} is a clique on $i$ vertices and a $(< i)$-clique is a clique on less than $i$ vertices.
For any graph $G$ and subset of vertices $X \subseteq V(G)$, we denote $G-X$ the graph whose vertex set is $V(G) \setminus X$ and edge set is $\{e \in E(G) : e \cap X = \emptyset \}$.
Let $H$ be a graph. We say that $G$ is \emph{$H$-free} if $G$
does not contain $H$ as induced subgraph.

\subsection*{Hypergraphs and covers} A \emph{hypergraph} is simply a collection of sets. So $|\cD|$ refers to the number of sets in the hypergraph $\cD$ and we define $V(\cD) = \bigcup_{D\in \cD} D$.
By $\cK_r(G)$ (resp.\ $\cK_{<r}(G)$) we denote the hypergraph of $r$-cliques (resp.\ $(<r)$-cliques) of $G$, i.e.\ $\cK_r(G) = \{X\subseteq V(G),\ X\ \text{is an $r$-clique}\}$.

A \emph{cover} of $\cD$ is a subset $X\subseteq V(\cD)$ that intersects every hyperedge of $\cD$. A \emph{matching} of $\cD$ is a collection of disjoint hyperedges.
The maximum size of a matching in $\cD$ is denoted by $\nu(\cD)$. Note that a cover of $\cD$ has always size at least $\nu(\cD)$ as it needs to intersect each of the elements of a maximum matching, which are disjoint.

Special cases of covers of hypergraphs are the covers of subgraphs of a graph.
For $G,H$ two graphs, an \emph{$H$-cover} in $G$ is a subset
$X\subseteq V(G)$ such that $G-X$ is $H$-free. In other words it is a cover of the hypergraph of the (induced) subgraphs of $G$ isomorphic to $H$.
In the $H$-\cov{} problem, given a graph $G$ and an integer $k \in \N$, one has to decide whether $G$ has a $H$-cover of size $k$.

\section{Dealing with large cliques}
\label{sec:covers}

A $K_r$-\cov{} can be useful to detect large cliques, as we explain now. This will allow us to identify cliques on which to branch.
\begin{lemma}\label{lem:bkeps}
  Given a graph $G$, a $K_r$-cover $M$, and an integer $p> r$, one can find a $p$-clique of $G$, or correctly conclude none exists, in $O(p^2|M|^p n^{r-1})$ steps.
\end{lemma}

\begin{proof}
  For each choice of $p - r + 1$ vertices of $M$ and $r-1$ other vertices of $G$ we check whether they form a clique (which takes $O(p^2)$ time), in which case return it and stop. If no clique is found we return that $G$ is $K_p$-free.
  The correctness follows from the following observation: as $G-M$ is
  $K_r$-free, every $p$-clique of $G$ has at least $p - r + 1$ of its
  vertices in~$M$.
\end{proof}

As noted in the proof of \autoref{lem:bkeps}, every large clique of the input graph will have most of its vertices in a $K_r$-cover so we can branch on which these are.

\begin{lemma}\label{lem:cliques}
There is an algorithm that, given two integers $r\geq 1$ and $p>r$, an instance $(G, k)$ of $K_r$-\cov{} and a $K_r$-cover $M$ of $G$, runs in $2^{kr(\log p)/p} |M|^p n^{O(r)}$ steps and returns a collection $\cY$ of $2^{kr(\log p)/p} \cdot |M|^p$ instances of the same problem such that:
\begin{enumerate}
     \item \label{e:equiv} $(G, k)$ is a positive instance if and only if $\mathcal{Y}$ contains a positive instance;
    \item \label{e:noclique} for every $(G', k') \in \cY$, $G'$ has no $p$-clique.
\end{enumerate}
\end{lemma}

\begin{proof}
    The algorithm is the following:
    \begin{enumerate}
      \item \label{e:base} Using the algorithm of \autoref{lem:bkeps} on $G$ and $M$, we find a $p$-clique $K$ (if none is found return $\cY = \{(G, k)\}$).
      \item We initialize $\cY = \emptyset$.
      \item Observe that any solution contains at least $p - r + 1$ vertices from $K$. For every subset $X$ of $K$ with $p - r + 1\leq |X| \leq k$ vertices:
      \begin{enumerate}
          \item \label{e:rec} We look for solutions $S$ that contain $X$ with a recursive call on $(G-X, k-|X|)$ and $M\setminus X$;
          \item We then add the resulting collection of instances to $\cY$.
      \end{enumerate}
      \item Return $\cY$.
    \end{enumerate}
    
    Regarding correctness, Item \eqref{e:noclique} follows from the base case of the recursion (step \ref{e:base}) and Item \eqref{e:equiv} can be proved by induction on the number of recursive calls, using the observation that $(G,k)$ has a solution containing a set $X$ if and only if $(G-X, k-|X|)$ has a solution.
    
    We denote by $T_{r,p}(n,k)$ the time complexity of the above algorithm with parameters $r,p, (G, k)$ where $|G| = n$.
    The time taken by each computation step is the following:
    \begin{enumerate}
        \item Finding a $p$-clique takes time $O(p^2
        |M|^pn^{r-1})$ by \autoref{lem:bkeps}.
        \item When a $p$-clique $K$ is found we consider at most $p^{r-1}$ subsets $X$ on which we perform a recursive call of cost at most $T_{r,p}(n-p+r-1, k-p+r-1)$.
    \end{enumerate}
    
    So
    \[
    T_{r,p}(n,k) \in |M|^p n^{O(r)} + p^{r-1} \cdot T_{r,p}(n-p+r-1, k-p+r-1).
    \]
    We deduce 
    \begin{align*}
        T_{r,p}(n,k) &\in p^{k(r-1)/(p-r+1)} |M|^pn^{O(r)}\\
        &\in 2^{kr(\log p)/p} |M|^p n^{O(r)},
    \end{align*}
    as desired. Note that a similar recurrence can be used to bound the number of output instances.
\end{proof}

\section{Picking petals}
\label{sec:petals}

Let $i\in \intv{1}{r-1}$ and let $X$ be an $i$-clique in a graph $G$.
An \emph{$r$-petal} of $X$ is a subset of vertices of $G-X$ that together with $X$ forms an $r$-clique. We denote by $\ext_r(G,X)$ the hypergraph of $r$-petals of $X$, i.e., $\ext_r(G,X) = \{Y \subseteq V(G)\setminus X,\ X\cup Y \in \cK_r(G)\}$. See \autoref{fig:petals} for an illustration.

\begin{figure}
    \centering
    \includegraphics{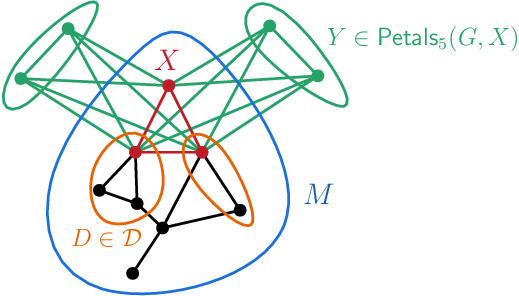}
    \caption{A $3$-clique $X$ with two $5$-petals that live in $G-M$. Here $\mathcal D$ contains two hyperedges, represented in orange. Observe that no hyperedge of $\mathcal D$ is contained in $X$, so $X$ is a lush $3$-clique.}
    \label{fig:petals}
\end{figure}

In order to deal more easily with the recursive steps in our algorithms we introduce \ann{}-$K_r$-\cov{}, an annotated version of $K_r$-\cov{} where a number of choices have already been made, which is recorded by extra vertex subsets that the solution is required to cover.

In the \ann{}-$K_r$-\cov{} problem, one is given a triple $(G, \cD, k)$ where $G$ is a graph, $\cD\subseteq \cK_{<r}(G)$ and $k$ is an integer. A \emph{solution} to this instance is a set of vertices that covers $\cK_r(G)$ and $\cD$ and has at most $k$ vertices. The question is whether the input instance admits a solution.

In the forthcoming algorithms it will also be more convenient to consider, together with an instance, a non-optimal solution $M$. This motivates the following definition.
A \emph{context} is a pair $((G,\cD,k), M)$ where $(G, \cD, k)$ is an instance of \ann{}-$K_r$-\cov{}, $M$ is a $K_r$-cover (possibly larger than $k$), and $V(\cD)\subseteq M$. We say that a context is \emph{positive} if the instance of \ann{}-$K_r$-\cov{} it contains is, and \emph{negative} otherwise.

Given a context $((G,\cD,k),M)$ and $i\in \intv{1}{r-1}$, a \emph{lush $i$-clique} is an $i$-clique $X$ of $G[M]$ that has an $r$-petal in $G-M$ and such that no $D\in \cD$ is subset of $V(X)$. See \autoref{fig:petals} for an example.

Informally, if $X$ is a lush clique then we are not guaranteed that covering $\cD$ alone does always also cover the cliques induced by $X$ and its petals, so we have to take care of them separately. Ideally we would like to get rid of lush cliques so that we can focus on $\cD$ to solve the problem.
We say that the context $((G,\cD,k), M)$ is \emph{$i$-stripped} if for every $i'<i$ it has no lush $i'$-clique. These notions are motivated by the following easy lemma.

\begin{lemma}\label{lem:kernel}
Let $((G, \cD)$ be an $r$-stripped context.
The instances $(G, \cD, k)$ and $(G[M], \cD, k)$ of \ann{}-$K_r$-\cov{} are equivalent.
\end{lemma}
\begin{proof}
First, recall that as $((G, \cD, k), M)$ is a context, $V(\cD) \subseteq M$ so $(G[M], \cD, k)$ is indeed a valid instance of \ann{}-$K_r$-\cov{}. Also as $\cK_r(G[M]) \subseteq \cK_r(G)$, if $(G, \cD, k)$ has a solution then $(G[M], \cD, k)$ does. So we only need to show the other direction.
Let $v \in V(G)\setminus M$ and suppose that $G$ contains some $r$-clique $X$ with $v\in X$.
The set $M\cap X$ is not empty because $M$ is a $K_r$-cover. Note that $X \cap M$ has a petal $X\setminus M$ disjoint from $M$. It cannot form a lush $i$-clique (for $i = |M\cap X|<r$) as this would contradict the assumption that the considered context is $r$-stripped, so there is some $D\in \cD$ that is subset of $X\cap M$.
Therefore every solution of $(G[M], \cD, k)$ does cover $X$ in $G$, as it covers $D$. As this holds for every $v$ and $X$ as above, every solution of $(G[M], \cD, k)$ is a solution of $(G, \cD, k)$, as desired.
\end{proof}

By the above lemma, if we manage to get rid of lush cliques, we can obtain an equivalent instance whose graph is not larger than $G[M]$.
As we will see, we are able to handle lush cliques to the price of slightly increasing the size of $M$ and producing several instances to represent the solutions of the original instance.

\begin{lemma}\label{lem:lushcliques}
  There is an algorithm that, given an integer $i<r$ and an $i$-stripped context $((G,\cD,k), M)$, runs in time $n^{O(r)}$ and either correctly concludes that the context is $(i+1)$-stripped, or returns a lush $i$-clique~$X$.
\end{lemma}

\begin{proof}
The input context is $i$-stripped so we only have to check whether it contains a lush $i$-clique.
We iterate over the $i$-cliques of $G[M]$. For every such clique $X$, we first check if $D\subseteq V(X)$ for some $D\in \cD$. If so $X$ is not a lush $i$-clique so we can move to the next choice of $X$. Otherwise we check if the common neighborhood of the vertices of $X$ in $G-M$ has an $(r-i)$-clique. If so this is an $r$-petal so $X$ is lush and we can return it. Otherwise we continue to the next choice of $X$. If the iteration terminates without detecting a lush $i$-clique, we can safely return that the input context is $(i+1)$-stripped.
The complexity bound follows from the fact that $|\cD| = n^{O(r)}$ and that $G[M]$ has $n^{O(r)}$ $i$-cliques.
\end{proof}

The following lemma is the key branching step in our subexponential algorithms for $K_r$-\cov{}.

\begin{lemma}\label{lem:flower-picking}
  There is an algorithm that, given $i \in \intv{1}{r-1}$ and $\lambda\in \intv{1}{k}$ and an $i$-stripped context $((G,\cD,k), M)$ where we denote by $\zeta$ the number of $i$-cliques in $G[M]$, runs in time $ 2^{O((k/\lambda) \cdot \log \zeta)} \cdot n^{O(r)}$
  and returns a collection $\mathcal{Z}$ of size $2^{O((k/\lambda) \cdot \log \zeta)}$ of 
  $(i+1)$-stripped contexts such that:
  \begin{enumerate}
      \item \label{e:posepos} $((G, \cD, k), M)$ is a positive context if and only $\cZ$ contains one; and
      \item \label{e:M} for every $((G', \cD', k'), M') \in \cZ$, $M\subseteq M'$, $|M'|\leq |M| + r(\lambda \zeta +k)$, and $k'\leq k$.
  \end{enumerate}
\end{lemma}

\begin{proof}
Let us define an auxiliary algorithm that takes as input $(i,\lambda,(G,\cD,k), M, \cP^*)$, where $(i,\lambda,(G,\cD,k), M)$ is as specified in the statement of the lemma, and $\cP^*$ is a matching of $\cD$, and output the promised collection $\cZ$. 
Such an auxiliary algorithm will imply the lemma, as we will run it with parameters $(i,\lambda,(G,\cD,k), M, \emptyset)$.
For simplicity, as $i$ and $\lambda$ will not change between recursive calls: we denote by $((G,\cD,k), M, \cP^*)$ the parameter of this auxiliary algorithm, 
which is defined as follows:
  \begin{enumerate}
   \item \label{e:tropcher} If $k < |\cP^*|$, we can immediately return $\cZ=\emptyset$.
    \item \label{e:stripped} Run the algorithm of \autoref{lem:lushcliques} on $i$ and $((G, \cD, k), M)$. If no lush $i$-clique is found then the input context is already $(i+1)$-stripped so we return $\cZ = \{((G, \cD, k), M)\}$. Otherwise let $X$ denote the lush $i$-clique we found.
    \item Construct the hypergraph $\cB_X$ of those $r$-petals of $X$ that are subset of $V(G)\setminus M$.
    \item \label{e:greeeedy} Greedily compute a maximal matching $\cP$ in $\cB_X$.
    \item \label{e:mincr} If $\tilde{\nu}_X \leq \lambda$, we return the result of the recursive call with parameters $((G, \cD, k), M\cup V(\cP),\cP^*)$  (and quit).
    \item\label{e:vroom} Otherwise, we investigate the different ways to cover the $r$-cliques induced by $X$ and $\cB_X$ (via $X$ or via $\cB_X$) as follows.
    \begin{enumerate}
      \item \label{e:center} Solutions containing some (yet unspecified) vertex of $X$: this is done by a recursive call with parameters $((G, \cD \cup \{X\}, k), M,\cP^*)$. We call $\cZ_1$ the resulting family.
      \item \label{e:petals} Solutions covering $\cP$.\footnote{Notice that some of these solutions have possibly already been investigated in the previous step. For our purpose it is not an issue however to consider several times the same solution.} As $\cP$ is disjoint from $M$ (hence from $\cP^*$), such solutions exist only if $k \geq |\cP^*| + |\cP|$. In this case we define $\mathcal{Z}_2$ as the result of the recursive call with parameters
        \[
            ((G, \cD\cup \cP, k), M\cup V(\cP), \cP^* \cup \cP).
        \]
        
        Otherwise we set $\mathcal{Z}_2 = \emptyset$.
        \end{enumerate}
        \item \label{e:basecase} We return $\mathcal{Z}_1 \cup \mathcal{Z}_2$.
  \end{enumerate}

Observe first that in step \eqref{e:petals}, the last parameter $\cP^* \cup \cP$ is a matching as required, as in particular $V(\cP) \subseteq V(G)-M$ and $V(\cP^*) \subseteq M$, by definition.
  We first prove the following fragment of item~\eqref{e:M}.
  \begin{claim}\label{cl:Mgonfle}
  For every $((G', \cD', k'), M')\in \cZ$, $M\subseteq M'$ and $k'\leq k$. 
  \end{claim}
  \begin{proof}
  The only places where $M$ is updated are steps~\ref{e:mincr} and \ref{e:petals}, where new vertices are added to it. Besides we never change the value of the parameter $k$.
  \end{proof}
  
  Let us describe the recursion tree $T$ of the above algorithm on some input $((G, \cD, k), M,\cP)$.
  The nodes of this tree are inputs. The root is $((G, \cD, k), M,\emptyset)$ and a node $s'$ is child of a node $s$ if a call of the above algorithm on the input $s$ triggers a call on the input $s'$. So the leaves of this tree are the inputs that do not trigger any recursive call.
  
  Let us consider a path from the root of $T$ to some leaf. We denote by
  \[
  ((G_j, \cD_j, k_j), M_j,\cP^*_j)_{j\in \intv{1}{\ell}}
  \]
  the inputs along this path, and by $C_j = (G_j, \cD_j, k_j), M_j)$ 
  the corresponding contexts,  with
  $((G_1, \cD_1, k_1), M_1,\cP^*_1)= ((G, \cD, k), M,\emptyset)$
  and $((G_\ell, \cD_\ell, k_\ell), M_\ell,\cP_\ell^*)$ corresponding to the aforementioned leaf.
  Also, for every $j\in \intv{1}{\ell-1}$ we denote by $X_j$ the lush $i$-clique of $C_j$ that is considered in the corresponding call.
  
  We first show that all lush $i$-cliques considered along this path belong to the cover $M=M_1$ of the initial context.
  
  \begin{claim}\label{cl:sameM}
  For every $j \in \intv{1}{\ell-1}$, $X_j \subseteq M$. 
  \end{claim}
  \begin{proof}
    Suppose towards a contradiction that for some $j$, $X_j \nsubseteq M$. Observe that since $C_j$ is not a leaf, $\cB_{X_j}$ is not empty so $X_j$ induces an $r$-clique together with some $r$-petal $B\in \cB_{X_j}$. Recall that $B$ is disjoint from 
    $M_j$, by definition of $\cB_{X_j}$. As $M_j$ is a superset of $M$ (\autoref{cl:Mgonfle}), $B$ is disjoint from $M$ as well. So $M$ intersects $X_j$ otherwise the $r$-clique $X_j\cup B$ would not be covered by $M$.
    Let $i'= |X_j\cap M|$. Note that $X_j\cap M$ is a lush $i'$-clique of $C_1$ since it has an $r$-petal $B\cup (X\setminus M)$ disjoint from $M$ (the fact that no set of $\cD$ is subset of $X_j\cap M$ follows from the fact that this property holds for $X_j$).
    By our initial assumption and as $|X_j| = i$, we have $i'<i$.  This contradicts the fact that $C_1$ is $i$-stripped.
  \end{proof}
  
  With a similar proof we can show the following (so the recursive calls are indeed made on valid inputs).
  \begin{claim}
  For every $i\in \intv{1}{\ell}$, $C_i$ is $i$-stripped.
  \end{claim}
  
  Let us now show that each lush clique is only considered once.
  \begin{claim}\label{cl:norepeat}
  For every distinct $j,j'\in \intv{1}{\ell-1}$, $X_j\neq X_{j'}$.
  \end{claim}
  \begin{proof}
  Suppose towards a contradiction that for some $j<j'$ we have $X_j = X_{j'}$.
  Then in the call on the context $C_j$, the next context $C_{j+1}$ was obtained at step~\ref{e:mincr} or~\ref{e:petals} (since at step~\ref{e:center} we would include $X_j$ in $\cD_j$, preventing it to be considered in future calls). In any of these two possibilities we set $M_{j+1} = M_j \cup V(\cP_j)$, where $\cP_j$ denotes the maximal matching of step~\ref{e:greeeedy} in the call on context $C_j$.
  
  Besides, as $X_{j}$ is a lush $i$-clique in the context $C_{j'}$, then it has some $r$-petal $B$ subset of $V(G_{j'})-M_{j'}$.
  As $M_{j+1} \supseteq M_{j'}$ (\autoref{cl:Mgonfle}), $B$ is also an $r$-petal of $X_j$ and is subset of $V(G_j)-M_j$ and by the above observation, it is disjoint from $\cP_j$. This contradicts the maximality of~$\cP_j$.
  \end{proof}
  
  Recall that the number of $i$-cliques in $M$ is $\zeta$. As a consequence of \autoref{cl:sameM} and \autoref{cl:norepeat}, we get the following.
  \begin{claim}\label{cl:recdepth}
  The recursion tree has depth at most $\zeta$.
  \end{claim}
  
  We can now conclude the proof of item~\eqref{e:M}.
  
  \begin{claim}\label{cl:Mbound}
    For every $((G', \cD', k'), M')\in \cZ$, $|M'| \leq |M| + r(\lambda\zeta + k)$.
  \end{claim}
  \begin{proof}
  When considering the context $C_j$, and given the chosen maximal matching $\cP_j$ of the petals of $X_j$, the set $M_{j+1}$ is defined from $M_j$ by:
  \begin{itemize}
      \item either adding the at most $\lambda(r-1)$ new vertices of $\cP_j$, if we make the recursive call at step~\ref{e:mincr},
      \item or by adding the vertices of the petals of $\cP_j$, if we recurse at step~\ref{e:petals}. Recall that in this case we also have $\cD_{j+1} = \cD_j \cup \cP_j$ and $\cP_{j+1}^* = \cP^*_j \cup \cP_j$
  \end{itemize}
  In the later case the number of added vertices is not directly bounded however we have
  $|\cP_{j+1}^*| = |\cP_{j}^*| + |\cP_j|$.
  Because of the stopping condition of step~\ref{e:tropcher}, we will overall (from $C_1$ to $C_\ell$) add at most $k$ petals to the hypergraph and each petal has at most $r-1$ vertices. Hence we get 
  $|M'| \leq |M| + r(\lambda\zeta + k)$, as claimed.
  \end{proof}

  Let us now show that the algorithm is correct, i.e., item~\ref{e:posepos} of the statement of the lemma.
  The proof is by induction on the depth of the recursion tree (i.e., $\ell-1$ with the notation above). When the depth is 0, there is no recursive call. This corresponds to the two base cases in this algorithms: step~\ref{e:tropcher}, when the ``budget'' $k$ is insufficient, and step~\ref{e:stripped}, when the input context is already $(i+1)$-stripped.
  Clearly the outputs in these cases satisfy~\ref{e:posepos}. 
  
  So we now consider the case of a run of the algorithm where the recursion tree has depth at least 1 and suppose that item~\ref{e:posepos} holds for all runs with recursion trees of smaller depth.
  \begin{itemize}
      \item If the recursive call is made at step~\ref{e:mincr} then 
  item~\ref{e:posepos} trivially holds because the instance is unchanged.
    \item Otherwise, note that any solution has to cover the $r$-cliques induced by $X$ and $\cB_X$.
    So any solution either contains a vertex of $X$, or covers $\cB_X$ (or both). These are exactly the two branches that are explored in steps~\ref{e:center} and~\ref{e:petals}, respectively, by our induction hypothesis.
  \end{itemize}
  
  The above shows that the algorithm is correct. It remains to prove that it has the claimed running time.
  Note that we do not update the graph neither the parameter between recursive calls: we will always work on the graph $G$ with $n$ vertices and with the parameter $k$. So the induction proving the time bound will use two different parameters as we explain now.
  For every $x,y \in \mathbb{N}$, let us denote by $T(x,y)$ the worst-case running time of the above algorithm on an input $((G, \cD, k), M,\cP^*)$ with $|G|=n$ such that there is at most $x$ (non-necessarily disjoint) lush $i$-cliques in $G[M]$ and such that $|\cP^*| \ge y$.
 Let $S_r(n)$ be the sum of the worst-case complexity of all subroutines needed in the different steps of item~\eqref{e:tropcher} to item~\eqref{e:basecase} to the exception of recursive calls (one such subroutine is the algorithm of  \autoref{lem:lushcliques}, another one is the construction of the hypergraph $\cB_X$).
 Observe that $T(x,y) \le S_r(n)$ when $y > k$ (as we fall into base case  of step~\ref{e:tropcher}) or when $x=0$ (as we fall into base case  of step~\ref{e:stripped}).
 Notice also that by definition we have $T(x,y) \le T(x',y)$ for any $x \le x'$, and $T(x,y) \le T(x,y')$ for any $y' \le y$.

 If we make a recursive call at step~\ref{e:mincr}, we return in time at most $T(x-1,y)$ (by induction) as, by \autoref{cl:norepeat}, no lush $i$-clique of the original instance is considered twice.
 Otherwise, we will make recursive calls in step~\ref{e:vroom} which by induction take time at most $T(x-1,y)+T(x-1,y+|\cP|)$ as, in the first branch \ref{e:center} of recursion, 
 \autoref{cl:norepeat} implies again that no lush $i$-clique is considered two times, and in the second branch~\ref{e:petals}, we know in addition that the size of the matching given as parameter increases by  $|\cP|$.
 Thus, in both cases (and including the other computation steps which take time $S_r(n)$), we obtain the upper bound:
\begin{align*}
    T(x,y) &\leq  T(x-1,y)+T(x-1,y+|\cP|)+S_r(n)  \\
    &\leq  T(x-1,y)+T(x,y+\lambda)+S_r(n).
\end{align*}
(For the last line recall that $T$ is anti-monotone with respect to its second parameter.)

Let us now show that for every $x,y\in \N$, $T(x,y) \leq x T(x, y+\lambda) + (x+1)S_r(n).$ The proof is by induction on $x$. The base case $x=0$ holds as observed above.
Suppose the inequality holds for $x-1$. As proved above
\begin{align}
T(x,y) &\leq T(x-1,y)+T(x,y+\lambda)+S_r(n)\nonumber\\
&\leq (x-1)T(x-1,y+\lambda) + xS_r(n)&\text{(by induction)}\nonumber\\
&\quad + T(x,y+\lambda)+S_r(n)\nonumber\\
&  \leq x T(x, y+\lambda) + (x+1)S_r(n), &\text{as claimed.}\label{eq:leq}
\end{align}
We now prove that for every $x,y\in \N$, 
\[
T(x,y) \leq x^{\frac{k+1-y}{\lambda}} + \left( 1+\frac{k+1-y}{\lambda}\right )(x+1) S_r(n).
\]
This time the induction is on $y$. The base case $y>k$ hold as observed above. Let $x\geq 1$ and $y\leq k$ and suppose the inequality holds for any pair $(x', y')$ with $y'>y$. Then as proved above in Eq.~\ref{eq:leq},
\begin{align*}
T(x,y) &\leq x T(x,y+\lambda)+(x+1)S_r(n)\\
&\leq x\cdot x^{\frac{k+1-y-\lambda}{\lambda}} + \left ( 1 + \frac{k+1- y - \lambda}{\lambda}\right )(x+1) S_r(n) + (x+1)S_r(n)\\
&\leq x^{\frac{k+1-y}{\lambda}} + \left( 1+\frac{k+1-y}{\lambda}\right )(x+1) S_r(n).
\end{align*}

Observe that $S_r(n)$ is dominated by the time spent in the algorithm of \autoref{lem:lushcliques}, and thus $S_r(n)\in n^{O(r)}$.
As we initially have $x \leq \zeta$ and $y \ge 0$,  we obtained the claimed running time.
A similar analysis can be used to bound the size of the output family $\cZ$.
  \end{proof}

By iterating the algorithm of \autoref{lem:flower-picking} for increasing values of $i$ we can obtain a collection of $r$-stripped contexts, as we explain now.
\begin{lemma}[picking petals]\label{lem:flowerS-picking}
 There is an algorithm that, given $i \in \intv{1}{r}$, $\lambda\in \R_{\geq 1}$, a context $((G,\cD,k),M)$, where we denote by $\zeta$ the number of $(<r)$-cliques in $G[M]$, runs in time 
 $ 2^{O((i \cdot k/\lambda) \cdot \log \zeta)} \cdot n^{O(r)}$
  and returns a collection $\mathcal{Z}$ of size 
  $ 2^{O((i \cdot k/\lambda) \cdot \log \zeta)}$
  of $i$-stripped contexts such that:
  \begin{enumerate}
      \item \label{e:posepos2} $(G, \cD, k)$ is a positive instance if and only if some input of $\cZ$ contains one; and
      \item \label{e:M2} for every $((G', \cD', k'), M') \in \cZ$, $M\subseteq M'$, $|M'|\leq |M| + ir(\lambda\zeta +k)$, and $k'\leq k$.
  \end{enumerate}
\end{lemma}

\begin{proof}
Again for the sake of clarity we assume $\lambda$ is a fixed constant.

The proof is by induction on $i$.
For the base case $i=1$ we simply observe that $((G, \cD, k), M)$ is already 1-stripped so there is nothing to do.

So let us now suppose that $i>1$ and that the statement holds for $i-1$.
So from the input context $((G, \cD, k), M)$ and $i-1$ we can use the induction hypothesis to construct a collection $\cZ_{i-1}$ of $(i-1)$-stripped contexts satisfying the statement for $i-1$.
Now  we apply the algorithm of \autoref{lem:flower-picking} to $i-1$ and each context in $\cZ_{i-1}$ and call $\cZ_i$ the union of the obtained collections.
Item~\ref{e:posepos2} follows from the properties of $\cZ_{i-1}$ and the correctness of the algorithm of \autoref{lem:flower-picking}.
Constructing $\cZ_{i-1}$ takes time
$ 2^{O(((i-1) \cdot k/\lambda) \cdot \log \zeta)} \cdot n^{O(r)}$ (by induction)
 and then we run the
$\left( 2^{O((k/\lambda) \cdot \log \zeta)} \cdot n^{O(r)}\right)$
-time algorithm of \autoref{lem:flower-picking} on each of its
 $ 2^{O(((i-1) \cdot k/\lambda) \cdot \log \zeta)}$
contexts. This results in an overall running time of 
$ 2^{O((i \cdot k/\lambda) \cdot \log \zeta)} \cdot n^{O(r)}$,
 as claimed. In each context of $\cZ_{i-1}$ the set $M'$ has size at most $|M| + (i-1)r(\lambda \zeta + k)$ (by induction) and after the run of the algorithm \autoref{lem:flower-picking}, the corresponding set in the produced instances has at most $r(\lambda \zeta + k)$ vertices more so we get the desired bound. Finally, as in the proof of \autoref{lem:flower-picking} the value of $k$ never changes.
\end{proof}

\section{Kick the cliques}
\label{sec:kick}

For every $\phi, \gamma\in \R_{\geq 0}$ and $\alpha \in (0,1)$, we say that graph class $\mathcal{G}$ has property $P_r(\phi, \gamma, \alpha)$ if there are functions $f(x) \in O(x^\phi)$ and $g(x) \in O(x^\gamma)$ such that for $G\in \mathcal{G}$ with $n$ vertices and clique number less than $\omega$,
\begin{enumerate}[(P1)]
    \item\label{p:cliques} $G$ has at most $f(\omega)\cdot n$ cliques of order less than $r$; and
    \item\label{p:tw} $\tw(G) \leq g(\omega) n^\alpha$.
\end{enumerate}
Our main contribution is the following.

\begin{theorem}\label{th:kick}
    For every hereditary graph class $\mathcal{G}$ that has property $P_r(\phi, \gamma, \alpha)$ for some $\alpha \in (0,1)$, $\phi, \gamma\in \R_{\geq 0}$,
    there is an algorithm that solves $K_r$-\cov{} on $\mathcal{G}$ in time
    \begin{align*}
    2^{
       O_{r,\phi}\left(k^{\eps} \log k\right )}
            \cdot n^{O_r(1)}\qquad \text{with}\quad \eps = \frac{\gamma + \alpha(\phi+2)}{\gamma + \alpha(\phi+1)+1}<1,
\end{align*}
    i.e., subexponential FPT time.
\end{theorem}

\begin{proof}
Let $f\colon x \mapsto c_f \cdot x^\phi$ and $g\colon x\mapsto c_g \cdot x^\gamma$ (for some $c_f,c_g>0$) be as in the definition of property $P_r(\phi, \gamma, \alpha)$.
For the sake of readability we will here use (when deemed useful) $\exp$ to denote the function $x \mapsto 2^x$ defined over reals.
Given an instance $(G,k)$ of $K_r$-\cov{}, we consider the context $((G,\cD, k), M)$ where $\cD = \emptyset$ and $M$ is an $r$-approximation of a $K_r$-cover computed in $n^{O(r)}$ time by greedily packing disjoint $r$-cliques.
If $|M| > kr$ we can already answer negatively, so in what follows we suppose that $|M|\leq kr$.
We run the algorithm of \autoref{lem:cliques} with $p=\ceil{k^{\eps}}$ for some constant $\eps\in (0,1)$ that we will fix later.
In time
\[
2^{O(rk^{1-\eps} \log k)} n^{O(r)}
\]
we obtain a set $\cY$ of $2^{O(rk^{1-\eps} \log k)}$ contexts that have no $p$-clique.
For each such context we apply the petal-picking algorithm of \autoref{lem:flowerS-picking} with
\begin{align*}
\lambda &= k^{\eps}\qquad \text{and}\\
\zeta &= f(p)\cdot |M|\\
&\in O\left (r k^{1+ \eps \phi}\right ) .
\end{align*}
Let $\cZ$ denote the union of the outputs families of these algorithms.
Computing this set then takes time
\begin{align*}
|\cY|\cdot 2^{O((r \cdot k/\lambda) \cdot \log \zeta)} \cdot n^{O(r)}
&\in 2^{O(r \cdot k^{1-\eps} (\log k + \log \zeta))} \cdot n^{O(r)} \\
&\in 2^{O(r^2 \cdot (1+\eps \phi)k^{1-\eps} \log k)} \cdot n^{O(r)}.
\end{align*}

For every $((G', \cD', k), M')\in \cZ$ we have
\begin{align*}
    |M'| &\leq |M| + r^2(\lambda \zeta + k) & \text{by \autoref{lem:flowerS-picking}}\\
    &\in O\left ( kr + r^2(rk^{1 + \eps\phi + \eps} + k) \right)\\
    &\in O\left (r^3k^{1 + \eps(\phi +1)} \right).
\end{align*}
So given any such context, we can decide whether it is positive or not in time
\begin{align*}
2^{g(p)|M'|^\alpha} n^{O(1)}
&\in \exp \left (O
    \left (
        p^\gamma \cdot \left (r^3 k^{1+\eps(\phi+1)} \right )^\alpha 
    \right ) \right )\cdot n^{O_r(1)}\\
&\in \exp
    \left (
        O\left (
            r^{3\alpha} \cdot k^{\eps \gamma + \alpha(1+\eps(\phi+1))}
        \right )
    \right ) \cdot n^{O_r(1)}
\end{align*}
as follows: first we use \autoref{lem:kernel} to delete irrelevant vertices and obtain an equivalent instance $H$ on $|M'|$ vertices, then we use property $P_r(\phi, \gamma, \alpha)$ to bound the treewidth of $H$ and then we solve the problem by dynamic programming on an approximate tree-decomposition in $2^{O(\tw(H))} n^{O(1)}$ time by noting that every $r$-clique and every $D'\in \cD'$ (which is an $(<r)$-clique) has to be contained in a bag.

So overall, computing $\cZ$ and solving the problem in each sub-instance takes time
\begin{align*}
    & \exp\left (
        O_r\left(
            (1+\eps \phi)k^{1-\eps} \log k
            + k^{\eps \gamma + \alpha(1+\eps(\phi+1))}
            \right )\right )
            \cdot n^{O_r(1)}
            \\
            \in &\exp\left (
        O_r\left(
            (1+\eps \phi) k^{1-\eps} \log k
            + k^{\eps(\gamma + \alpha + \alpha\phi) + \alpha}
            \right )\right )
            \cdot n^{O_r(1)}.
\end{align*}

As we aim for algorithms where the contribution of $k$ to the time complexity is of the form $2^{o(k)}$, the above bound sets the following constraint: $\eps<\frac{1-\alpha}{\gamma+\alpha+\alpha\phi}$.
Let $\eps = \frac{1-\alpha-\delta}{\gamma+\alpha+\alpha\phi}$ for some constant $\delta\in (0,1)$ that we will fix later. Then the above complexity becomes:
\begin{align*}
\exp\left (
        O_r\left(
            (1+\eps \phi)k^{1-\eps} \log k
            + k^{1-\delta}
            \right )\right )
            \cdot n^{O_r(1)}.
\end{align*}
We optimize (ignoring logarithmic factors) by choosing the value of $\delta$ so that $1-\eps= 1-\delta$, i.e., $\delta = \frac{1-\alpha}{\gamma + \alpha(\phi+1) +1}$. This gives the following overall time bound:
\begin{align*}
2^{
        O_r\left((1+\eps \phi)k^{\eps'} \log k\right )}
            \cdot n^{O_r(1)}\qquad \text{with}\ \eps' = \frac{\gamma + \alpha(\phi+2)}{\gamma + \alpha(\phi+1)+1}.
\end{align*}

As $\alpha<1$ we have $\eps'<1$ so the algorithm runs in subexponential FPT time, as desired.
\end{proof}

\section{Applications}
\label{sec:apps}

In this section we give applications of \autoref{th:kick} to specific graph classes. First we have to show that the considered classes satisfy property $P_r$ (recall that this property is defined at the beginning of \autoref{sec:kick}).

\subsection{Pseudo-disk graphs and map graphs}
In this subsection we will prove the following lemma.

\begin{lemma}\label{lem:pdg}
 Pseudo-disk graphs and map graphs have the property $P_r(r-2, 1/2, 1/2)$.
\end{lemma}
As a consequence we get the two following results.

\thpseudo*

\thmap*

To prove \autoref{lem:pdg}, we first need to state some external results. For $d\in \N$ we say that a graph $G$ is \emph{$d$-degenerate} if every subgraph of $G$ (including $G$ itself) has a vertex of degree at most $d$.
In order to bound the number of small cliques in the considered graphs we can bound their degeneracy and then rely on the following result of Chiba and Nishizeki.

\begin{theorem}[\cite{chiba1985arboricity}]\label{th:chiba}
 Any string graph $G$ with $n$ vertices and degeneracy $d$ has $O(id^{i-1} n)$ $i$-cliques.
\end{theorem}
Actually \cite{chiba1985arboricity} gives a time bound for the enumeration of $i$-cliques in graphs of arboricity~$d$. As arboricity and degeneracy are linearly bounded by each other and since the time bound implies a bound on the number of enumerated objects (up to a constant factor), we get the above statement. 

\begin{theorem}[\cite{berthe24bFVS}]\label{th:vomi}
Pseudo-disk graphs on $n$ vertices with clique number $\omega$ have at most $3e\omega n$ edges and treewidth $O(\sqrt{\omega n})$. In particular they are $(3e\omega)$-degenerate.
\end{theorem}

\begin{theorem}[\cite{chen2002map}]\label{th:mapedges}
Map graphs on $n$ vertices with clique number $\omega$ have at most $7\omega n$ edges. In particular they are $(7\omega)$-degenerate.
\end{theorem}
To bound the treewidth of map graphs we use the following combination of results on balanced separators of string graphs of Lee and the links between separators and treewidth of Dvo{\v{r}}{\'a}k and Norin.

\begin{theorem}[\cite{lee2016separators} and  \cite{DVORAK2019137}] \label{th:sep-m-tw}
    Any $m$-edge string graph has treewidth $O(\sqrt{m})$.
\end{theorem}
As a consequence of \autoref{th:mapedges} and \autoref{th:sep-m-tw} we get the following.
\begin{corollary}\label{cor:maptw}
Map graphs on $n$ vertices with clique number $\omega$ have treewidth $O(\sqrt{\omega n})$.
\end{corollary}

\begin{proof}[Proof of \autoref{lem:pdg}]
Let $G$ be a pseudo-disk graph with $n$ vertices and clique number $\omega$.
By \autoref{th:vomi} the pseudo-disk graphs with $n$ vertices and clique number $\omega$ are $(3e\omega)$-degenerate, so by \autoref{th:chiba} they have $O((r-1)^2 (3e\omega)^{r-2} n)$ cliques of order less than $r$. So property \Pref{p:cliques} holds with $\phi = r-2$.
By \autoref{th:vomi} property \Pref{p:tw} is satisfied with $\alpha=1/2$ and
$\gamma=1/2$.
The proof for map graphs is very similar, using \autoref{th:mapedges} and \autoref{cor:maptw}.
\end{proof}

\subsection{String graphs}
We now move to string graphs where, as discussed in the introduction, forbidding large bicliques is necessary. Actually for $K_{t,t}$-subgraph free string graphs the branching of \autoref{lem:bkeps} to reduce the clique number is not necessary in the algorithm of \autoref{th:kick} as the number of small cliques and the treewidth are already suitably bounded. This explains the zeroes in \autoref{lem:string} hereafter.  In this subsection we will prove the following lemma.

\begin{lemma}\label{lem:string}
 $K_{t,t}$-subgraph-free string graphs have the property $P_r(0,0,1/2)$.
\end{lemma}

As a consequence we get the following result.

\begin{corollary}\label{cor:string}
There is an algorithm solving $K_r$-\cov{} in $K_{t,t}$-subgraph-free string graphs in time 
\[
2^{O_{t,r}(k^{2/3}\log k)} \cdot n^{O_r(1)}.
\]
\end{corollary}
The contributions of $t$ and $r$ to the complexity in \autoref{cor:string} are not explicit due to the way we stated the bound of \autoref{th:kick} but can be extracted from the proof of this theorem.

\begin{theorem}[\cite{lee2016separators}]\label{th:lee}
For every $t\in \N$, $K_{t,t}$-subgraph-free string graphs on $n$ vertices have degeneracy $O(t \log t)$ and treewidth $O(\sqrt{n\cdot t \log t})$.
\end{theorem}
\begin{proof}[Proof of \autoref{lem:string}.]
Combining \autoref{th:chiba} and \autoref{th:lee} we get that the number of $(<r)$-cliques in $K_{t,t}$-subgraph-free string graphs is $O_r((t\log t)^{r-2} n)$, i.e.  \Pref{p:cliques} holds with $\phi=0$.
\autoref{th:lee} gives \Pref{p:tw} with $\alpha=1/2$ and $\gamma=0$.
\end{proof}

\subsection{Minor-closed classes}

 In this subsection we will prove the following lemma.

\begin{lemma}\label{lem:hminfree}
For every graph $H$, $H$-minor-free graphs have the property $P_r(0,0,1/2)$.
\end{lemma}

As a consequence we get the following result.

\thminorfree*

To prove \autoref{lem:hminfree}, we first need to state some external results.

\begin{theorem}[see \cite{wood2007maximum}]\label{th:wood}
Every $d$-degenerate graph with $n\geq d$ vertices has at most $2^d(n-d+1)$ cliques.
\end{theorem}

\begin{theorem}[\cite{kostochka1986minimum, thomason1984extremal}]\label{th:kosto}
For every $h$-vertex graph $H$ there is a constant $d = O(h \sqrt{\log h})$ such that $H$-minor-free graphs are $d$-degenerate.
\end{theorem}

\begin{theorem}[\cite{Alon1990separator}]\label{th:twsqrt}
For every graph $H$, $n$-vertex $H$-minor-free graphs have treewidth $O_H(\sqrt{n})$.
\end{theorem}

\begin{proof}[Proof of \autoref{lem:hminfree}]
As a consequence of \autoref{th:kosto} and \autoref{th:wood}, $H$-minor-free graphs have a linear number of cliques (regardless of their clique number).
Hence they satisfy \Pref{p:cliques} with $\phi=0$. By \autoref{th:twsqrt} they also satisfy \Pref{p:tw} with $\alpha=1/2$ and $\gamma=0$.
\end{proof}

\section{Open problems}\label{sec:open}
As discussed in the introduction, our main result provides a generic way to obtain subexponential parameterized algorithms for $K_r$-\cov{}, which can in particular be applied to several graph classes for which such algorithms were know from previous work (for \trh{}, the special case $r=3$). Nevertheless there is still a gap between the running times of these different applications of our algorithm and the best time bounds for these specific classes. One can for instance compare our \autoref{th:pseudo} (resp.\ \autoref{th:minorfree}) with the previous results corresponding to items \ref{e:oh} and~\ref{e:contact} of \autoref{th:pos} (resp.\ \autoref{th:apexfree}). It would be nice to match these known bounds, or to improve them when possible. More generally we can ask about the infimum $\eps$ such that $K_r$-\cov{} can be solved in time $2^{O(k^\eps)}n^{O(1)}$ in the classes we considered. We recall that under ETH, $K_3$-\cov{} cannot be solved in time $2^{o(\sqrt{n})}$ (so $\eps\geq 1/2$) even for a very restricted subclass of string graphs (\autoref{th:neg}).

A second research direction is to understand for which graphs $H$ our results can extended to the $H$-\cov{} problem (where one wants to hit any subgraph isomorphic to $H$).
In disk graphs for example, it is already known \cite{lokSODA22} that there exists subexponential FPT algorithms for $P_\ell$-\cov{} when $\ell \le 5$.

Recall that in this paper we gave sufficient conditions for a hereditary graph class to admit a subexponential FPT algorithm for $K_r$-\cov{}. It remains an open problem to characterize such classes.

\bigskip
\newcommand{\etalchar}[1]{$^{#1}$}

\end{document}